\documentclass[conference,a4paper]{IEEEtran}
\usepackage[utf8]{inputenc}
\usepackage[T1]{fontenc}

\usepackage[cmex10]{amsmath}
\usepackage{mleftright}
\mleftright
\usepackage{graphicx}
\usepackage{url}
\usepackage{booktabs}

\usepackage{amssymb}
\usepackage{physics}
\usepackage{amsthm}

\usepackage{cite}
\usepackage{amsfonts}
\usepackage[noend]{algorithmic}
\usepackage{textcomp}
\usepackage{xcolor}
\def\BibTeX{{\rm s\kern-.05em{\sc i\kern-.025em s}\kern-.08em
    T\kern-.1667em\lower.7ex\hbox{E}\kern-.125emX}}

\usepackage{algorithm}
\usepackage[caption=false]{subfig}

\usepackage{charter, latexsym}
\usepackage{enumerate}

\usepackage{easyReview}

\newtheorem{lemma}{Lemma}

\newtheorem{theorem}{Theorem}

\usepackage{graphicx}
\usepackage[export]{adjustbox}

\begin{document}

\title{Minimizing Moments of AoI for Both Active and Passive Users through Second-Order Analysis}

\author{
    \IEEEauthorblockN{Siqi Fan}
    \IEEEauthorblockA{
    \textit{Texas A\&M University}\\
    College Station, USA \\
    sf26372@tamu.edu}
    \and
    \IEEEauthorblockN{Yuxin Zhong\textsuperscript{\textsection}}
    \IEEEauthorblockA{
    \textit{Georgia Institute of Technology}\\
    Atlanta, USA \\
    yzhong332@gatech.edu}
    \and
    \IEEEauthorblockN{I-Hong Hou}
    \IEEEauthorblockA{
    \textit{Texas A\&M University}\\
    College Station, USA \\
    ihou@tamu.edu}
    \and
    \IEEEauthorblockN{Clement Kam}
    \IEEEauthorblockA{%
    \textit{U.S. Naval Research Laboratory}\\
    Washington, USA \\
    ckk@ieee.org}
}

\maketitle

\begingroup\renewcommand\thefootnote{\textsection}
\footnotetext{Yuxin’s work is done in her study in Texas A\&M university.}

\begin{abstract}


In this paper, we address the optimization problem of moments of Age of Information (AoI) for active and passive users in a random access network. In this network, active users broadcast sensing data while passive users only receive signals. Collisions occur when multiple active users transmit simultaneously, and passive users are unable to receive signals while any active user is transmitting. Each active user follows a Markov process for their transmissions. We aim to minimize the weighted sum of any moments of AoI for both active and passive users in this network. To achieve this, we employ a second-order analysis to analyze the system. Specifically, we characterize an active user's transmission Markov process by its mean and temporal process. We show that any moment of the AoI can be expressed a function of the mean and temporal variance, which, in turn, enables us to derive the optimal transmission Markov process. Our simulation results demonstrate that this proposed strategy outperforms other baseline policies that use different active user transmission models.

\end{abstract}

\section{Introduction}

Recently, there has been a significant increase in the use of time-sensitive applications such as Internet of Things, vehicular networks, sensor networks, and drone systems for surveillance and reconnaissance. While traditional research has focused on optimizing the reliability and throughput of communication, these efforts often fall short in meeting the demands for fresh data. To address this, a performance metric known as age-of-information (AoI) was introduced \cite{kaul2012real} and has received increasing attention in recent literature \cite{huang2015optimizing,chen2016age,kosta2017age,yates2017status,hsu2017age,kadota2018scheduling,yates2018status, sun2018age, sun2019closed,maatouk2020asymptotically,zou2021optimizing,jaiswal2021minimization,yates2020age,moltafet2020average,wang2022useful,inoue2017stationary,moltafet2022aoi,moltafet2022moment,banerjee2020fundamental,he2022decentralized,chen2020age,pan2022age,yavascan2021analysis,munari2021modern,saurav2021game,yang2021understanding,kriouile2021minimizing,chen2021timely,guo2022theory,yao2022age,maatouk2022analysis}.

Previous studies on AoI have commonly considered centralized control over transmission decisions \cite{yates2018status, sun2018age, kadota2018scheduling, sun2019closed, maatouk2020asymptotically, zou2021optimizing,jaiswal2021minimization} and queuing problems without transmission collisions \cite{moltafet2020average,moltafet2022aoi,wang2022useful,moltafet2022moment,inoue2017stationary}. However, centralized policies and collisions can lead to large overhead or extra communication, which can be a significant drawback in delay-sensitive applications such as the Internet of Things and transmission-sensitive applications like satellite communications.

To address this issue, we consider a random access network, where active users broadcast sensed data and multiple transmissions at the same time can cause collisions. Since there are no acknowledgements in broadcast, active users do not have any feedback information on collisions. This makes recent feedback-based algorithms for AoI optimization \cite{chen2020age,pan2022age,yavascan2021analysis,munari2021modern,he2022decentralized,saurav2021game} inapplicable.

Additionally, we consider that the network has silent observers, referred to as passive users, who share the same channel as active users. Passive users aim to observe out-of-network radio activities in the same channel but do not make any transmissions. Examples of passive users include sensors detecting malicious jammers, receivers of satellite communications, and radio telescopes tracking celestial objects. The coexistence of passive and active users can improve channel utilization, as passive users can obtain valuable observations when no active user is transmitting.

As AoI depends on previous successful transmissions, memoryless random access algorithms, such as slotted ALOHA, are not suitable for optimizing AoI. Therefore, we consider that each active user follows a Markov process to determine its transmission activities. Furthermore, instead of only considering the mean of AoI, we consider the moments of AoI, which can offer insights into the variance and the maximum value of AoI.

The goal of this paper is to minimize the weighted sum of any moment of AoI for both active users and passive users. There are several challenges that need to be addressed. First, since the transmission strategy of an active user is Markovian instead of i.i.d., the temporal correlation of its transmission activities needs to be explicitly taken into account. Second, active users and passive users have different behavior, which leads to different definitions of AoI between the two types of users. Third, most existing studies on moments of AoI, such as \cite{inoue2017stationary,moltafet2022aoi, moltafet2022moment}, focus only on simple queueing systems and are not applicable to random access networks. Finally, most existing studies on AoI in random access networks only focus on the first moment of AoI and can only obtain the optimal solution in the asymptotic sense.


To address above challenges and solve the optimization problem, we propose to investigate the system behavior through a second-order analysis. We analyze the mean and temporal variance of delivery process to characterize the system. Then, we formulate the expressions of moments of AoI based on these mean and temporal variance, and then further formulate them as functions of state change probabilities in Markov process. By revealing special properties in expressions of moments of AoI for both active and passive users, we find that the optimal action for an active user, under some minor conditions, is to become silent immediately after one attempt of transmission. The simulation results show that our formulation of moments of AoI is valid and our solution outperforms other baseline algorithms that have different active user activity models.

The rest of the paper is organized as follows. Section~\ref{sec:system} introduces our system model and the optimization problem. Section~\ref{sec:model} formulates expressions of moments of AoI based for both active and passive users based on a second-order model. Section~\ref{sec:opt} solves the optimal transmission strategies. Some simulation results are presented in Section~\ref{sec:sim}. Finally, Section~\ref{sec:conclusion} concludes the paper.

\section{System setting}
\label{sec:system}

We consider a wireless system with one channel and two types of users, active users and passive users. Each active user monitors a particular field of interest and uses the wireless channel to transmit its surveillance information to its neighboring receivers. Passive users do not make any wireless transmissions and instead monitor the radio activities in the wireless channel. For example, in a battlefield scenario, an active user can be a surveillance drone that monitors a certain area in the battlefield and broadcasts its video feed to all nearby units. A passive user can be a signal detector that aims to identify and locate enemy jammers or communication devices operating in the same channel.

Next, we discuss the interference relationship among users. We assume that active users are grouped into $C$ clusters and each cluster has $N$ ($N\geq 2$) active users. Active users in the same cluster interfere with each other but do not suffer from interference from other clusters. An active user can successfully deliver a packet if it is the only device in the cluster that is transmitting. If multiple active users in the same cluster transmit at the same time, then a collision occurs and none of the packets are received. On the other hand, since passive users need to detect enemy devices that are potentially far away, we assume that a passive user is interfered by all active users and it can only detect radio activities if no active user is transmitting.

The performance of each user is determined by its associated AoI. We assume that time is slotted and the duration of one time slot is the time needed to transmit one packet. Each active user generates a new surveillance packet in each slot. Hence, the AoI of an active user $n$ at time $t$ is defined by following recursion:
\begin{align*}
    AoI_{n,t} = \begin{cases}
                    1, & \text{if users $n$ successfully delivers} \\
                    & \text{a packet at time $t$,} \\
                    AoI_{n,t-1} + 1, & \text{otherwise}.
                \end{cases}
\end{align*}
Since passive users can only detect radio activities when no active users are transmitting, the AoI of passive user $j$ at time $t$ is defined by
\begin{align*}
    AoI_{j,t} = \begin{cases}
                    1, \; & \text{if no active user transmits,} \\
                    AoI_{j,t-1} + 1, \; & \text{otherwise}.
                \end{cases}
\end{align*}

We now discuss the transmission strategy of each active user. We assume that the strategy of an active user is governed by a Markov process with two states, a TX state and an Idle state. An active user will transmit a packet at time $t$ if and only if it is in the TX state. If an active user is in the TX state at time $t$, then its state will change to Idle at time $t+1$ with probability $s$. If an active user is in the Idle state at time $t$, then its state will change to TX at time $t+1$ with probability $r$. Fig.~\ref{fig:state} presents this Markov process. We note that this strategy generalizes the ALOHA network. In particular, if $r+s=1$, then an active user will transmit in each time slot with probability $r$ independently.

\begin{figure}[thb]
    \centering
    \includegraphics[scale=0.23]{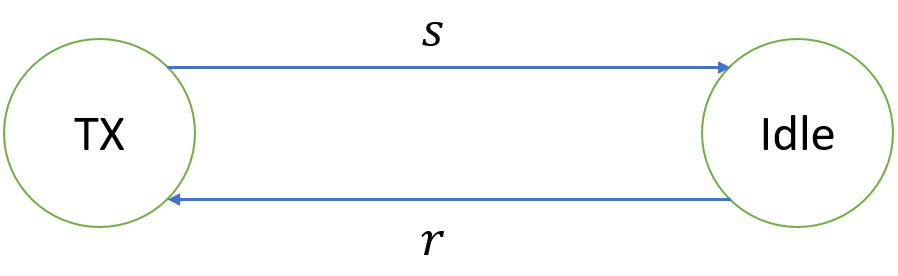}
    \caption{Markov process model.}
    \label{fig:state}
\end{figure}

We aim to minimize the moments of AoI for both active users and passive users. Specifically, let $AoI_a$ be the random variable of an active user's AoI and let $AoI_p$ be the random variable of a passive user's AoI in steady state. For a fixed integer $z$ and a fixed real number $w\in[0,1]$, we aim to find $r$ and $s$ that solve the following optimization problem:
\begin{align}
    \min \; & F(r,s) := w E[AoI_a^z] + (1-w) E[AoI_p^z] \label{opt_problem}\\
    \text{s.t.} \; & r,s\in [0,1], \label{restriction}
\end{align}
where $E[\cdot]$ is the expectation function.

\section{The Second Order Model for AoI}
\label{sec:model}

A key challenge in solving the optimization problem (\ref{opt_problem})-(\ref{restriction}) is that it is hard to express $E[AoI_a^z]$ and $E[AoI_p^z]$ as functions of $r$ and $s$. We propose to adopt the framework of second-order network optimization in Guo \textit{et al.} \cite{guo2022theory} to address this challenge. The framework in \cite{guo2022theory} is only applicable to centrally scheduled systems and the mean of AoI. In this section, we will extend this framework to model random access networks and moments of AoI.

We first define the second-order model that describes the performance of both active and passive users. For an active user $n$, Let $S_n(t)$ be the indicator function that $n$ successfully delivers a packet at time $t$, that is, $n$ is the only active user in its cluster that transmits at time $t$. We call $[S_n(1), S_n(2), \dots]$ the \emph{delivery process} of $n$. Due to symmetry, the delivery processes of all active users follow the same distribution. We then define the mean and temporal variance of the delivery process of an active user by
\begin{align*}
    m_a &:= \lim_{T\xrightarrow{} \infty } \frac{\sum_{t=1}^{T}S_n(t)}{T}, \\
    v_a^2 &:= E[(\lim_{T\xrightarrow{} \infty } \frac{\sum_{t=1}^{T}S_n(t)-T m_a}{\sqrt{T}})^2],
\end{align*}
respectively, where the expectation is taken over all sample paths.

For a passive user $j$, we let $S_j(t)$ be the indicator function that no active user transmits, and hence the passive user can monitor radio activities without interference, at time $t$. We call $[S_j(1), S_j(2), \dots]$ the \emph{passive detection process}. The mean and temporal variance of the passive detection process are defined as $m_p := \lim_{T\xrightarrow{} \infty } \frac{\sum_{t=1}^{T}S_j(t)}{T}$ and $v_p^2 := E[(\lim_{T\xrightarrow{} \infty } \frac{\sum_{t=1}^{T}S_j(t)-T m_p}{\sqrt{T}})^2]$, respectively.

The second-order model uses the means and temporal variances to capture and approximate all aspects of the system. To show its utility, we first derive the closed form expressions of $m_a$, $m_p$, $v_a^2$, and $v_p^2$. Let $\lambda := \frac{r}{r+s}$ and $\theta = 1-r-s$, we have the following theorem: 

\begin{theorem}
    \begin{align*}
        m_a &= \lambda (1-\lambda)^{N-1}, \quad m_p = (1-\lambda)^{CN}, \\
        v_a^2 =&  2\sum_{k=1}^{\infty}((\lambda+(1-\lambda)\theta^{k}) (1-\lambda+\lambda \theta^{k})^{N-1} -m_a) m_a\nonumber\\
        & + m_a - m_a^2, \\
        v_p^2 = & 2\sum_{k=1}^{\infty}((1-\lambda +\lambda \theta^{k})^{CN} -m_p) m_p +  m_p - m_p^2.
    \end{align*}
    \label{thm:mean_vairance}
\end{theorem}
\begin{proof}
    See Appendix \ref{app:mean}.
\end{proof}

Next, we show that the moments of AoI of both active users and passive users can be approximated as functions of $(m_a, v_a^2)$ and $(m_p, v_p^2)$, respectively.

Due to space limitation, we focus on discussing the derivation of moments of AoI of active users.

Let $t_i$ be the $i$-th time that an active user $n$ successfully deliveries its packet, and define $l_n(i):=t_i-t_{i-1}$. For active user $n$, its AoI from $t_{i-1}$ to $t_i-1$ is from $1$ to $l_n(i)$. So, the summation of $z$-th moment of AoI from $t_{i-1}$ to $t_i-1$ is $\sum_{k=1}^{l_n(i)} k^z$. Let $L_n$ be the total number of deliveries for active user $n$. Then, summing over all $l_n(i)$, we have the sum of all $z$-th moment of AoI for active user $n$, $\sum_{i=1}^{L_n} \sum_{k=1}^{l_n(i)} k^z$. Next, divided by the whole time $\sum_{i=1}^{L_n} l_n(i)$ and choosing the expectation, we have $E[AoI_a^z] = E\Big[\lim_{L_n\xrightarrow{} \infty}\frac{\sum_{i=1}^{L_n}\sum_{k=1}^{l_n(i)}k^z}{\sum_{i=1}^{L_n}l_n(i)}\Big]$. By using the Faulhaber’s formula, we have 

\begin{align}
    E[AoI_a^z] & =  E\Big[\lim_{L_n\xrightarrow{} \infty}\frac{\sum_{i=1}^{L_n}1}{\sum_{i=1}^{L_n}l_n(i)}\Big(\frac{1}{\sum_{i=1}^{L_n}1}\sum_{i=1}^{L_n} \big( \frac{l_n(i)^{z+1}}{z+1} \nonumber  \\
    & \qquad + \frac{l_n(i)^z}{2} + \sum_{k=2}^{z} \frac{B_k}{k!}\frac{z!l_n(i)^{z-k+1}}{(z-k+1)!} \big) \Big)\Big] \nonumber \\
    & = \frac{1}{E[l_n(i)]}\big(\frac{E[l_n(i)^{z+1}]}{z+1} + \frac{E[l_n(i)^z]}{2} \nonumber \\
    & \qquad + \sum_{k=2}^z\frac{B_k}{k!}\frac{z!E[l_n(i)^{z-k+1}]}{(z-k+1)!} \big),
\end{align}
where $B_k$ is the Bernoulli number. 

To approximate $E[l_n(i)^{k}]$, we construct an alternative sequence $S'_n(t)$ that have the same mean and temporal variance as $S_n(t)$. Let $D_n(t)$ be a Brownian motion random process with mean $m_a$ and variance $v_a^2$. The sequence $\{S'_n(1), S'_n(2), \dots\}$ is defined as
\begin{align}
    S'_n(t) =  \begin{cases}
                    1, \; & \text{if $D_n(t)-D_n(t^-) \geq 1$,} \\
                    0, \; & \text{otherwise,}
                \end{cases}
    \label{eq:def_s'}
\end{align}
where $t^-=\max\{\tau|\tau<t, S'_n(\tau)=1\}$. Intuitively, we assume that a packet delivery occurs in the alternative sequence every time $D_n(t)$ increases by one since the last packet delivery in this sequence. Hence, $\sum_{\tau=1}^t S_n'(\tau)$ is close to $D_n(t)$, and the process $[S'_n(1), S'_n(2), ...]$ has mean $m_a$ and temporal variance $v_a^2$.

Let $l_n'(i):=t'_{n,i} - t'_{n,i-1}$, where $t'_{n,i}$ is the $i$-th time that $S'_n(t)=1$. Since $S'_n(t)$ has the same mean and temporal variance as $S_n(t)$, we propose to approximate $E[l_n(i)^k]$ by $E[l'_n(i)^k]$. Moreover, $l'_n(i)$ can be regarded as the time needed for $D_n(t)$ to increase by $1$, which is equivalent to the first-hitting time for the Brownian motion random process with level $1$. Thus, $l_n'(i)$ follows the inverse Gaussian distribution $IG(\frac{1}{m_a},\frac{1}{v_a^2})$ (\cite{schrodinger1915theorie,folks1978inverse}). By using the moment generating function of inverse Gaussian distribution in \cite{krishnamoorthy2006handbook}, we have $E[l'_n(i)^k] = \frac{1}{m_a^k}\sum_{\zeta=0}^{k-1}\frac{(k-1+\zeta)!}{\zeta!(k-1-\zeta)!}(\frac{2m_a}{v_a^2})^{-\zeta}$.

For a passive user $j$, let $\Bar{t}_i$ be the $i$-th time that $S_j(t)=1$ and $l_j(t):=\Bar{t}_i-\Bar{t}_{i-1}$. By combining the derivations above, we obtain the following:
\begin{theorem}
    Let $B_k$ be the Bernoulli number, then
    \begin{align}
        E[AoI_a^z] = & \frac{1}{E[l_n(i)]}(\frac{E[l_n(i)^{z+1}]}{z+1} + \frac{E[l_n(i)^z]}{2} \nonumber\\
        & + \sum_{k=2}^{z}\frac{B_k}{k!}\frac{z! E[l_n(i)^{z-k+1}]}{(z-k+1)!}),
        \label{eq:active_aoi}
    \end{align}
    and
    \begin{align}
        E[AoI_p^z] =& \frac{1}{E[l_j(i)]}(\frac{E[l_j(i)^{z+1}]}{z+1} + \frac{E[l_j(i)^z]}{2} \nonumber\\
        & + \sum_{k=2}^{z}\frac{B_k}{k!}\frac{z! E[l_j(i)^{z-k+1}]}{(z-k+1)!}),
    \end{align}
    where $E[l_n(i)^{k}] \approx \frac{1}{m_a^k}\sum_{\zeta=0}^{k-1}\frac{(k-1+\zeta)!}{\zeta!(k-1-\zeta)!}(\frac{2m_a}{v_a^2})^{-\zeta}$ and $E[l_j(i)^{k}] \approx \frac{1}{m_p^k}\sum_{\zeta=0}^{k-1}\frac{(k-1+\zeta)!}{\zeta!(k-1-\zeta)!}(\frac{2m_p}{v_p^2})^{-\zeta}$.
    \label{theorem:aoi_z}
    $\hfill\Box$
\end{theorem}
For example, we have $E[AoI_a] = \frac{1}{2}(\frac{v_a^2}{m_a^2}+\frac{1}{m_a})+\frac{1}{2}$ and $E[AoI_a^2] = \frac{v_a^4}{m_a^4}+\frac{v_a^2}{m_a^3}+\frac{3v_a^2+2}{6m_a^2}+\frac{1}{2m_a}+\frac{1}{6}$.

With the help of Theorem \ref{theorem:aoi_z}, $F(r,s)$ can now be approximated as a function of $(m_a, v_a^2, m_p, v_p^2)$, which, as stated in Theorem~\ref{thm:mean_vairance}, can be expressed as functions of $r$ and $s$.


\section{AoI Optimization in Second Order Model}
\label{sec:opt}

In the previous section, we approximate $E[AoI_a^z]$ and $E[AoI_p^z]$ as functions of $r$ and $s$. In this section, we further study the optimal behavior of active users and solve the optimal $r$ and $s$. We find the interesting property that, when $N>C+4$, the optimal solution requires an active user to the Idle state every time it makes a transmission, i.e., $s=1$. This result is formally stated below.
\begin{theorem}
    When $N>C+4$, there exists a $\lambda\in[0,\frac{1}{N}]$ such that choosing $r=\frac{\lambda}{1-\lambda}$ and $s=1$ minimizes $F(r,s)$.
    \label{theorem:opt_sol}
\end{theorem}
In addition to showing the interesting behavior that, under the optimal solution, an active user will never transmit in two consecutive slots, this theorem also simplifies the search for the optimal solution. In particular, it shows that one only needs to find the optimal $\lambda$, which can be done through a simple line search, in order to obtain the optimal solution.



Recall that $\lambda = \frac{r}{r+s}$ and $\theta=1-r-s$, or, equivalently, $r=\lambda(1-\theta)$ and $s=(1-\lambda)(1-\theta)$. Before proving Theorem 3, we first establish some properties about the optimal $\lambda$ and $\theta$. We start by the following lemma: 
\begin{lemma}
    For any positive integer $z$, $E[AoI_a^z]$ is an increasing function of $\frac{1}{m_a}$ and $\frac{v_a^2}{m_a^2}$, and $E[AoI_p^z]$ is an increasing function of $\frac{1}{m_p}$ and $\frac{v_p^2}{m_p^2}$.
    \label{lemma:incr}
\end{lemma}
\begin{proof}
    For active users, we note that each term in (\ref{eq:active_aoi}) is a constant multiplying $\frac{E[l_n(i)^k]}{E[l_n(i)]}$, where $k\geq 1$. For $\frac{E[l_n(i)^k]}{E[l_n(i)]}$, we have 
    \begin{align}
        \frac{E[l_n(i)^k]}{E[l_n(i)]} = \frac{1}{m_a^{k-1}}\sum_{\zeta=0}^{k-1}\frac{(k-1+\zeta)!}{\zeta!(k-1-\zeta)!}(\frac{2m_a}{v_a^2})^{-\zeta},
        \label{eq:temp_ratio}
    \end{align}
    whose $\zeta$-th term is $\frac{(k-1+\zeta)!}{2^{\zeta}\zeta!(k-1-\zeta)!} \frac{1}{m_a^{k-1-\zeta}} \frac{v_a^{2\zeta}}{m_a^{2\zeta}}$. Since $k-1\geq \zeta$, this $\zeta$-th term in (\ref{eq:temp_ratio}) is an increasing function of $\frac{1}{m_a}$ and $\frac{v_a^2}{m_a^2}$. Thus, $\frac{E[l_n(i)^k]}{E[l_n(i)]}$ is an increasing function of $\frac{1}{m_a}$ and $\frac{v_a^2}{m_a^2}$, for any $k\geq 1$. Therefore, $E[AoI_a^z]$ is an increasing function of $\frac{1}{m_a}$ and $\frac{v_a^2}{m_a^2}$.
    
    Similarly, $E[AoI_p^z]$ is an increasing function of $\frac{1}{m_p}$ and $\frac{v_p^2}{m_p^2}$.
\end{proof}

Based on Lemma ~\ref{lemma:incr}, we further have
\begin{lemma}
    For any $w\in[0,1]$, The optimal $\lambda$ that minimizes $F(r,s)$ is in range $[0,\frac{1}{N}]$.
    \label{lemma:opt_l}
\end{lemma}
\begin{proof}[Sketch of proof]

    Recall $m_a = \lambda (1-\lambda)^{N-1}$ and $m_p=(1-\lambda)^{CN}$. When $\lambda>\frac{1}{N}$, it is easy to verify that $m_a$ and $m_p$ decreases as $\lambda$ increases.
    
    In Appendix~\ref{app:opt_l}, we further show that, when $\lambda \geq \frac{1}{N}$, $\frac{\partial (v_a^2/m_a^2)}{\partial \lambda} \geq 0$ and $\frac{\partial (v_p^2/m_p^2)}{\partial \lambda} \geq 0$. 
    
    Therefore, when $\lambda\geq \frac{1}{N}$, ($\frac{1}{m_a}$, $\frac{v_a^2}{m_a^2}$, $\frac{1}{m_p}$, $\frac{v_p^2}{m_p^2})$ all increase as $\lambda$ increases. Since $E[AoI_a^z]$ and $E[AoI_p^z]$ are increasing functions of $(\frac{1}{m_a}, \frac{v_a^2}{m_a^2})$ and $(\frac{1}{m_p}, \frac{v_p^2}{m_p^2})$ respectively, the optimal $\lambda$ is less than or equal to $\frac{1}{N}$.
\end{proof}


Next, we analyze the optimal $\theta$. The lemma below derives the optimal $\theta$ when $\lambda$ is given and fixed.
\begin{lemma}
When $\lambda \in [0, \min\{\alpha,\beta\}]$ is fixed, setting $\theta=-\frac{\lambda}{1-\lambda}$, or, equivalently, setting $r = \frac{\lambda}{1-\lambda}$ and $s = 1$, minimizes $E[AoI_a^z]$ and $E[AoI_p^z]$, where $\alpha$ and $\beta$ are the smallest positive roots of
\begin{align*}
    h_N(y) &=-(N+8)y^3 - (N-13) y^2 - 6y + 1, \\
    \Bar{h}_{CN}(y) &= (CN-10)y^3 - (CN-13)y^2 - 6y + 1,
\end{align*}
respectively.
\label{lemma:opt_theta}
\end{lemma}
\begin{proof}[Sketch of proof]
    When $\lambda$ is determined, $m_a = \lambda (1-\lambda)^{N-1}$ and $m_p=(1-\lambda)^{CN}$ are determined. Hence, by Theorem~\ref{theorem:aoi_z}, $E[AoI_a^z]$ and $E[AoI_p^z]$ are increasing functions of $v_a^2$ and $v_p^2$ respectively.

    In Appendix~\ref{app:active}, we further establish that $\frac{\partial v_a^2}{\partial \theta} \geq 0$ when $\lambda\leq \alpha$, and $\frac{\partial v_p^2}{\partial \theta} \geq 0$ when $\lambda \leq \beta$. So, when $\lambda \in [0, \min\{\alpha,\beta\}]$, minimum $\theta$ will minimize $E[AoI_a^z]$ and $E[AoI_p^z]$. In Appendix~\ref{app:active}, we also show that $\alpha, \beta < \frac{1}{2}$, and hence the minimum $\theta$ is achieved when $r=\frac{\lambda}{1-\lambda}$ and $s=1$.

    Therefore, when $\lambda \in [0, \min\{\alpha,\beta\}]$ is fixed, choosing $r=\frac{\lambda}{1-\lambda}$ and $s=1$ minimizes $E[AoI_a^z]$ and $E[AoI_p^z]$.
\end{proof}

A limitation of Lemma~\ref{lemma:opt_theta} is that it requires $\lambda \in [0, \min\{\alpha,\beta\}]$. In the following lemma, we establish lower bounds of $\alpha$ and $\beta$.
\begin{lemma} 
$\alpha > \frac{1}{N}$ when $N > 4$, and $\beta > \frac{1}{N}$ when $N > C+4$.
\label{lemma_alpha}
\end{lemma}
\begin{proof}[Sketch of proof]
    In Appendix~\ref{app:lemma_a}, we prove that $h_N(y)$ and $\Bar{h}_{CN}(y)$ both have only one root in $[0,1]$. 
    
    We further show that $h_N(\frac{1}{N})>0$ when $N > 4$, and $\Bar{h}_{CN}(\frac{1}{N})>0$ when $N > C+4$, in Appendix~\ref{app:lemma_a}. Since $h_N(1)=-2N<0$ and $\Bar{h}_N(1) = -2<0$, we have $\alpha > \frac{1}{N}$ when $N > 4$, and $\beta > \frac{1}{N}$ when $N > C+4$.
\end{proof}

We are now ready to prove Theorem~\ref{theorem:opt_sol}.
\begin{proof}[Proof of Theorem~\ref{theorem:opt_sol}]
    It is shown in Lemma~\ref{lemma:opt_l} that the optimal $\lambda$ is in the range $[0,\frac{1}{N}]$, which is covered by the range $[0,\min\{\alpha,\beta\}]$ when $N>C+4$ according Lemma~\ref{lemma_alpha}. Therefore, the optimal solution in Lemma~\ref{lemma:opt_theta} holds for the optimal $\lambda$ when $N>C+4$. Hence, when $N>C+4$, choosing $r=\frac{\lambda}{1-\lambda}$ and $s=1$ minimizes $E[AoI_a^z]$ and $E[AoI_p^z]$, as well as $F(r,s)$.
\end{proof}



\section{Simulation Results}
\label{sec:sim}

This section presents simulation results for applying our solution of $r$ and $s$. The performance function is the optimization objective $F(r,s)$.

The results of our solution are compared with three other baseline policies. We provide descriptions of all four policies, along with necessary modifications to fit the simulation setting.
\begin{itemize}
    \item \textbf{Our solution}: Based on Theorem~\ref{theorem:opt_sol}, we choose the optimal $\lambda$ and set $r=\frac{\lambda}{1-\lambda}$ and $s=1$ for all active users. This optimal $\lambda$ is obtained by plugging $r=\frac{\lambda}{1-\lambda}$ and $s=1$ into Theorem~\ref{theorem:aoi_z} and doing an exhaustive search over $\lambda$ for $F(r,s)$ with a precision $0.01$.
    \item \textbf{Slotted ALOHA}: When $r+s=1$, our active user model is reduced to be slotted ALOHA, which is a famous and commonly used model for transmissions with no acknowledgment. According to solutions in \cite{yates2017status}, we choose $r=\frac{1}{N}$ and $s=1-\frac{1}{N}$. 
    \item \textbf{Optimal ALOHA}: After introducing passive users, the optimal $\lambda$ in slotted ALOHA is no longer $\frac{1}{N}$. Thus, we simulate all possible $\lambda$ with precision $0.01$, and then choose the best $\lambda$ to obtain optimal slotted ALOHA. It should be noticed that optimal ALOHA is impossible to implement in real world since its $\lambda$ is determined after the whole process is finished. 
    \item \textbf{Age Threshold ALOHA (ATA)}: Age threshold ALOHA is proposed by Yavascan and Uysal \cite{yavascan2021analysis}, which is a slotted ALOHA algorithm with an age threshold. If the AoI of one active user is larger than a threshold, this active user will follow slotted ALOHA transmission with probability $r$. Otherwise, this active user will stay silent. Since there is no acknowledgment in the system, we let active users assume their transmissions are always successful. As suggested in \cite{yavascan2021analysis}, $r=4.69/N$ and the age threshold is $2.2N$.
\end{itemize}

We simulate different scenarios with varying values of $w$, where $w$ ranges from $0$ to $1$. We set $N=7$, $C=2$ or $4$, and $z=1$ or $2$. Once $N$, $C$ and $z$ are determined, we can set the values of $r$ and $s$ based on policy designs. The system is then simulated for $10$ individual runs, each with $100,000$ time slots. The final results are the average of these $10$ runs. We present the value of $\frac{\text{Actual } F(r,s) \text{ of one algorithm}}{\text{Theoretical } F(r,s) \text{ of our solution}}$ for all four algorithms under each setting. Simulation results are shown in Fig.~\ref{fig:sim}. As some values are extremely large, only results that have ratio below $2.4$ are shown in figures.

\begin{figure}[thb]
    \subfloat[$C=2$, $z=1$.]{\includegraphics[width = 1.65in]{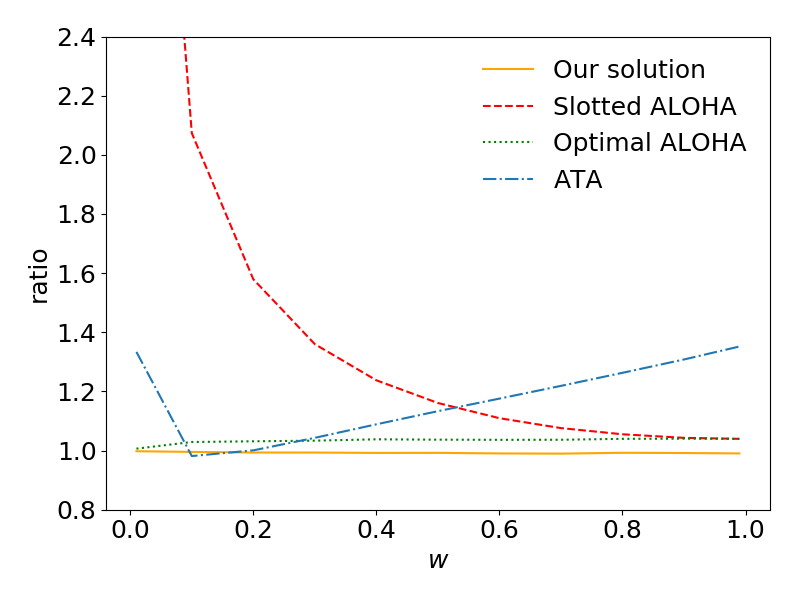}\label{fig:sim_1}}
    \subfloat[$C=2$, $z=2$.]{\includegraphics[width = 1.65in]{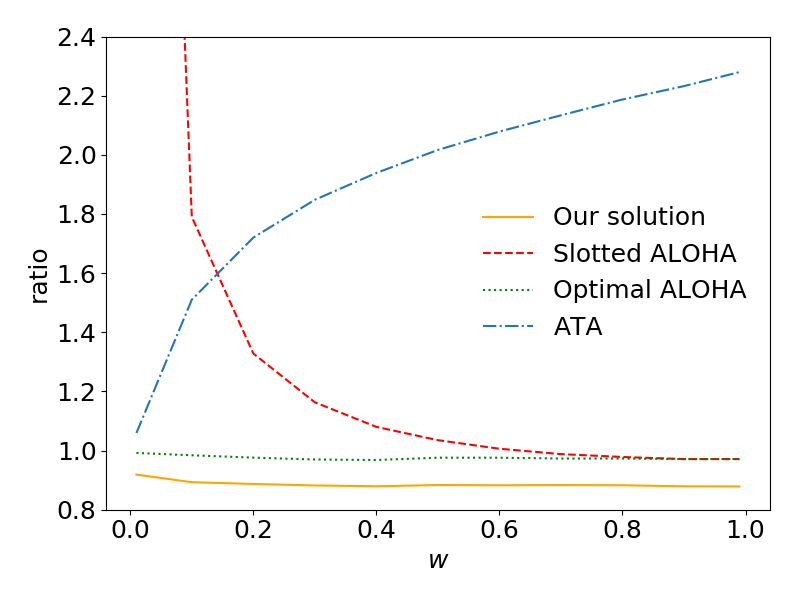}\label{fig:sim_2}}\\
    \subfloat[$C=4$, $z=1$.]{\includegraphics[width = 1.65in]{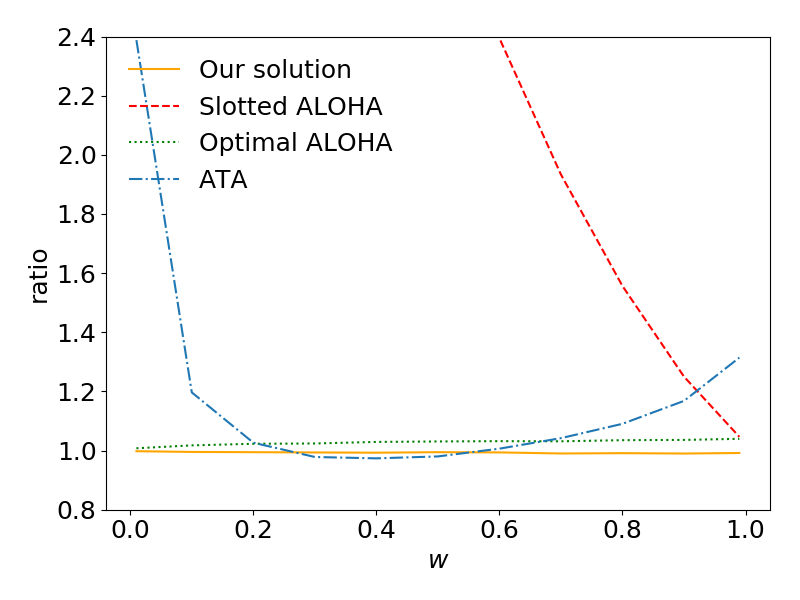}\label{fig:sim_3}}
    \subfloat[$C=4$, $z=2$.]{\includegraphics[width = 1.65in]{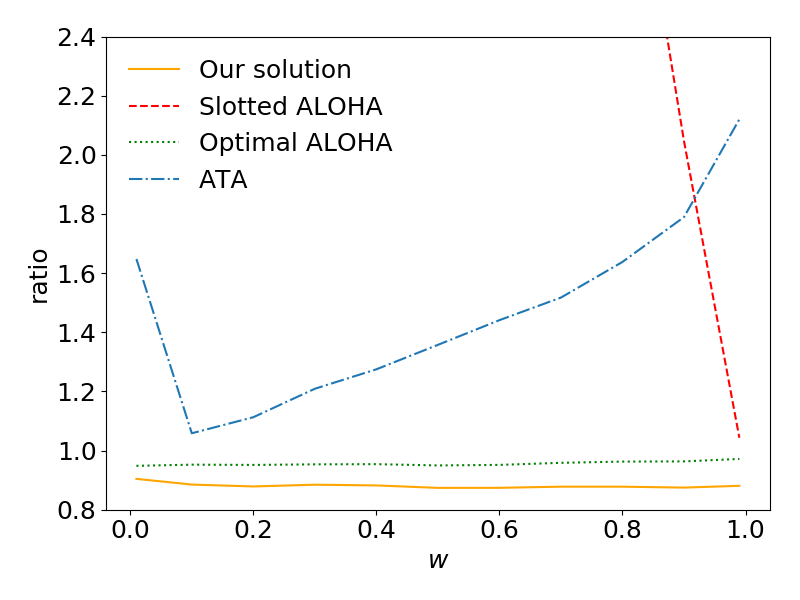}\label{fig:sim_4}}
    \caption{$\frac{\text{Actual } F(r,s) \text{ of one algorithm}}{\text{Theoretical } F(r,s) \text{ of our solution}}$  when $N=7$}
    \label{fig:sim}
\end{figure}

The results show that the average mismatches between theoretical $F(r,s)$ and actual $F(r,s)$ of our solution are about $0.7\%$ when $z=1$ and $12\%$ when $z=2$. Furthermore, we notice the actual value is always lower than theoretical value. These results indicate that our approximation is useful in practice.

Comparing the performance of different algorithms in Figure~\ref{fig:sim_1} and Figure~\ref{fig:sim_2}, we can observe that our strategy outperforms other three policies in all settings. Even when compared to the optimal ALOHA, which is not practical in real-world scenarios, our strategy shows a $4\%$ reduction when $z=1$ and a $9\%$ reduction when $z=2$. Additionally, the increase in the reduction of our strategy between $z=1$ and $z=2$ indicates that our strategy can better optimize the variance of AoI.

In Figure~\ref{fig:sim_3} and Figure~\ref{fig:sim_4}, we simulate $C=4$ to evaluate the performance of our strategy when the condition $N>C+4$ is not met. The Results indicate that our strategy outperforms the other three algorithms in nearly all settings, even when the condition $N>C+4$ is not satisfied.

\section{Conclusion}
\label{sec:conclusion}

This paper studies an optimization problem of moments of AoI for both active and passive users in a random access network. In this network, active users broadcast their transmissions following a Markov process and will interfere with each other. By applying a second-order model, we formulate moments of AoI as functions of state change probabilities in the Markov process. Next, we reveal the special properties of these functions and optimize moments of AoI for both active and passive users. The optimal solution indicates that the best strategy for active users is to become silent immediately after one transmission. Simulations show that this transmission strategy outperforms three baseline algorithms under various settings.

\section*{Acknowledgment}

This material is based upon work supported in part by NSF under Award Number ECCS-2127721, in part by ARL and ARO under Grant Number W911NF-22-1-0151, and in part by ONR under Contract N00014-21-1-2385.


\newpage

\IEEEtriggeratref{18}

\bibliographystyle{ieeetr}
\bibliography{reference}

\newpage

\appendix

\subsection{Proof of Theorem 1}
\label{app:mean}

    Let $x_{n,t}$ be the indicator that the active user $n$ is in TX state in time slot $t$.
    
    First, we calculate the mean for active user $n$ in cluster $c_i$ and passive user $j$. We have
    \begin{align*}
        m_a &= Prob(x_{n,t}=1)\prod_{\zeta\neq n, \zeta\in c_i}Prob(x_{\zeta,t}=0) \\
        &= \frac{r}{r + s}\Big(\frac{s}{r + s}\Big)^{N-1} = \lambda (1-\lambda)^{N-1},
    \end{align*}
    and $m_p = \prod_{\zeta=1}^{CN} Prob(x_{\zeta,t}=0) = (\frac{s}{r + s})^{CN} = (1-\lambda)^{CN}$.

    Then, we derive the temporal variance for active user $n$ and passive user $j$. According to the central limit theorem of Markov process, we have
    \begin{align*}
        v_a^2 &= Var(S_n(1)) + 2\sum_{k=1}^{\infty} Cov(S_n(1),S_n(k+1)), \\
        v_p^2 &= Var(S_j(1)) + 2\sum_{k=1}^{\infty} Cov(S_j(1),S_j(k+1)).
    \end{align*}
    
    We first calculate $v_a^2$. Since $S_n(1)$ is a Bernoulli random variable, it has
    \begin{align*}
         &Var(S_n(1))   = E[S_n(1)] - E[S_n(1)]^2 =  m_a - m_a^2 \\
         &Cov(S_n(1), S_n(k+1)) \\
         &= E[S_n(1)S_n(k+1)]-E[S_n(1)]E[S_n(k+1)] \\
         &=(Prob(S_n(k+1)=1|S_n(1)=1)-E[S_n(k+1)]) E[S_n(1)] \\
         & = (Prob(S_n(k+1)=1|S_n(1)=1) - m_a)m_a.
    \end{align*}
    Assume the active user $n$ is in cluster $c_i$, we have $Prob (S_n(k+1)=1|S_n(1)=1) =  Prob(x_{n,k+1}=1|S_n(1)=1)\prod_{\zeta\neq n, \zeta\in c_i} Prob(x_{\zeta,k+1}=0|S_n(1)=1)$.
    
    Let $G_{n,n}(k) = Prob(x_{n,k}=1|S_n(1)=1)$ and $G_{\zeta,n}(k) = Prob(x_{\zeta,k}=1|S_n(1)=1), l\neq n, l\in c_i$, we have $Prob (S_n(k)=1|S_n(1)=1) = G_{n,n}(k) \prod_{\zeta\neq n, \zeta\in c_i}(1-G_{\zeta,n}(k))$ and
    \begin{align*}
        v_a^2 =& 2\sum_{k=1}^{\infty} ( G_{n,n}(k+1) \prod_{\zeta\neq n}(1-G_{\zeta,n}(k+1)) -m_a) m_a \\
        &+ m_a - m_a^2.
    \end{align*}
    $G_{n,n}(k)$ can be calculated by
    \begin{align*}
        G_{n,n}(k) &= Prob(x_{n,k}=1|S_n(1)=1) \\
        &= G_{n,n}(k-1) Prob(x_{n,k}=1|x_{n,k-1}=1) \\
        &\quad + (1-G_{n,n}(k-1))Prob(x_{n,k}=1|x_{n,k-1}=0) \\
        &= G_{n,n}(k-1)(1-s) + (1-G_{n,n}(k-1))r \\
        &= r + \theta G_{n,n}(k-1),
    \end{align*}
    for $k>1$, and $G_{n,n}(1)=1$. By $\sum_{i=2}^{k}\theta^{k-i} G_{n,n}(i)$, we have
    \begin{align*}
        G_{n,n}(k) &= \sum_{i=2}^{k}\theta^{k-i}r + \theta^{k-1}G_{n,n}(1) \\
        &= \frac{r}{r+s}+\frac{s}{r+s}\theta^{k-1} = \lambda + (1-\lambda)\theta^{k-1}.
    \end{align*}
    
    For $G_{\zeta,n}(k)$, we also have
    \begin{align*}
        G_{\zeta,n}(k) &= Prob(x_{\zeta,k}=1|S_n(1)=1) \\
        &= G_{\zeta,n}(k-1) Prob(x_{\zeta,k}=1|x_{\zeta,k-1}=1) \\
        &\quad + (1-G_{\zeta,n}(k-1))Prob(x_{\zeta,k}=1|x_{\zeta,k-1}=0) \\
        &= G_{\zeta,n}(k-1)(1-s) + (1-G_{\zeta,n})(k-1)r \\
        &= r + \theta G_{\zeta,n}(k-1),
    \end{align*}
    for $k>1$, and $G_{\zeta,n}(1)=0$. Solving the recursive equation, we have $G_{\zeta,n}(k) = \lambda - \lambda \theta^{k-1}$. Therefore,
    \begin{align*}
        & v_a^2 =  2\sum_{k=1}^{\infty}((\lambda + (1-\lambda) \theta^{k})(1-\lambda+\lambda \theta^{k})^{N-1} -m_a) m_a \\
        &\quad + m_a - m_a^2.    
    \end{align*}
    
    Now, we calculate $v_p^2$. Sine $S_j(1)$ is also a Bernoulli random variable in the system, we have
    \begin{align*}
         &Var(S_j(1)) = E[S_j(1)] - E[S_j(1)]^2 = m_p - m_p^2 \\
         &Cov(S_j(1), S_j(k+1)) \\
         &= E[S_j(1)S_j(k+1)]-E[S_j(1)]E[S_j(k+1)] \\
        &=(Prob(S_j(k+1)=1|S_j(1)=1)-E[S_j(k+1)]) E[S_j(1)] \\
        & = (Prob(S_j(k+1)=1|S_j(1)=1) - m_p)m_p.
    \end{align*}
    
    Let $G_{n,j}(k) = Prob(x_{n,k}=1|S_j(1)=1)$, we have $Prob(S_j(k+1)=1|S_j(1)=1) = \prod_{n=1}^{CN}(1-G_{n,j}(k))$ and
    \begin{align*}
        v_p^2& = m_p- m_p^2 + 2\sum_{k=1}^{\infty} ( \prod_{n=1}^{CN}(1-G_{n,j}(k+1)) -m_p) m_p.
    \end{align*}
    $G_{n,j}(k)$ can be calculated by
    \begin{align*}
        G_{n,j}(k) &= Prob(x_{n,k}=1|S_j(1)=1) \\
        &= G_{n,j}(k-1) Prob(x_{n,k}=1|x_{n,k-1}=1) \\
        &\quad + (1-G_{n,j}(k-1))Prob(x_{n,k}=1|x_{n,k-1}=0) \\
        &= G_{n,j}(k-1)(1-s) + (1-G_{n,j})(k-1)r \\
        &= r + \theta G_{n,j}(k-1),
    \end{align*}
    for $k>1$, and $G_{n,j}(1)=0$. This expression is similar to that of $G_{\zeta,n}(k)$. Thus, solving the recursive equation, we have $G_{n,j}(k) = \lambda - \lambda \theta^{k-1}$. Thus,
    \begin{align*}
        & v_p^2 = 2\sum_{k=1}^{\infty}((1-\lambda + \lambda \theta^{k})^{CN} -m_p) m_p +  m_p - m_p^2.
    \end{align*}

\subsection{Proof of Lemma \ref{lemma:opt_l}}
\label{app:opt_l}
    First, we show that we only need to consider the case $\theta\leq 0$. The partial derivative of $\frac{v_a^2}{m_a}$ with respect to $\theta$ when $\lambda$ is fixed is
    \begin{align*}
         \frac{\partial v_a^2}{\partial \theta} &= \frac{2}{m_a}\sum_{k=1}^{\infty}[(1-\lambda)k\theta^{k-1}(1-\lambda+\lambda \theta^{k})^{N-1} \\
        & \quad +(\lambda+(1-\lambda)\theta^k)(N-1)\lambda k \theta^{k-1}(1-\lambda+\lambda \theta^k)^{N-2})] \\
        &= \frac{2}{m_a} \sum_{k=1}^{\infty}[1-2\lambda + N\lambda^2+ N\lambda(1-\lambda)\theta^{k}] \\
        & \quad \times k\theta^{k-1}(1-\lambda+\lambda \theta^{k})^{N-2}. 
    \end{align*}
    When $\theta>0$, it is easy to verify that the partial derivative is positive. Similarly, we also have the partial derivative of $\frac{v_p^2}{m_p}$ with respect to $\theta$ is positive when $\theta>0$. Thus, the optimal $\theta$ is non-positive, and hence we only consider $\theta\leq 0$ in the following proof.

    We first prove $\frac{\partial (v_a^2/m_a^2)}{\partial \lambda} \geq 0$ when $\lambda\geq \frac{1}{N}$. Based on Theorem~\ref{thm:mean_vairance}, we have
    \begin{align*}
        & \frac{\partial (v_a^2/m_a^2)}{\partial \lambda} = 2\sum_{k=1}^{\infty} \big[ \frac{(1-\theta^k)(1-\lambda+\lambda\theta^k)^{N-1}m_a}{m_a^2} \\
        & + \frac{(\lambda+(1-\lambda)\theta^k)(N-1)(1-\lambda+\lambda\theta^k)^{N-2}(\theta^k-1)m_a}{m_a^2} \\
        & -\frac{(1-N\lambda)(1-\lambda)^{N-2}(\lambda+(1-\lambda)\theta^k)(1-\lambda+\lambda\theta^k)^{N-1}}{m_a^2} \big]\\
        & -\frac{(1-N\lambda)(1-\lambda)^{N-2}}{m_a^2} \\
        & = \frac{(N\lambda-1)(1-\lambda)^{N-2}}{m_a^2} + 2\sum_{k=1}^{\infty} \frac{(1-\lambda)^{N-2}(1-\lambda+\lambda\theta^k)^{N-2}}{m_a^2} \\
        & \times \theta^k\big[ (N-2)\lambda^2(1-\lambda)\theta^k + (N-2)\lambda^2 +2\lambda -1 \big].
    \end{align*}
    Let $\omega(\lambda,\theta,k,N):=(1-\lambda)^{N-2}(1-\lambda+\lambda\theta^k)^{N-2}\theta^k[ (N-2)\lambda^2(1-\lambda)\theta^k + (N-2)\lambda^2 +2\lambda -1 ]$ be the $k$-th term in the above infinite summation. The following part is to prove $\sum_{k=1}^{\infty}\xi(\lambda,\theta,k,N) \geq 0$ or $\sum_{k=1}^{\infty}\xi(\lambda,\theta,k,N)+(N\lambda-1) \geq 0$. 
    
    Based on definition of $\lambda$ and $\theta$, we first have $(1-\lambda)\geq 0$ and $(1-\lambda+\lambda\theta^k)\geq 0$.

    Consider the special case $\theta=-1$. $\theta=-1$ can only be achieved when $\lambda=\frac{1}{2}$. Thus, $(1-\lambda+\lambda\theta^k)^{N-2}$ equals to $1$ when $k$ is even, and equals to $0$ when $k$ is odd. In this case, $(N-2)\lambda+2\lambda-1>0$. Thus, $\omega(\lambda,\theta,k,N) \geq 0$ for any $k$, and hence $\frac{\partial (v_a^2/m_a^2)}{\partial \lambda} \geq 0$.

    Now, we discuss the case $\theta > -1$.
    
    We first consider $(N-2)\lambda^2+2\lambda-1\geq 0$, which means $\lambda\in[\frac{\sqrt{N-1}-1}{N-2},1]$. In this case, $[ (N-2)\lambda^2(1-\lambda)\theta^k + (N-2)\lambda^2 +2\lambda -1 ]$ is non-negative for all even $k$, and is possible negative for odd $k$ (when $k$ is larger, equation becomes positive). 
    
    Consider $k$ from $2$ to infinite, we compare the value with an even $k=2q$ and its following odd $k=2q+1$. As $\theta \in (-1,0]$, we have $|(1-\lambda+\lambda\theta^{2q})^{N-2}|>|(1-\lambda+\lambda\theta^{2q+1})^{N-2}|$ and $|\theta^{2q}| > |\theta^{2q+1}|$. Since $\theta \leq 0$, when $k$ is even, we have $\theta^k\geq 0$ and $(N-2)\lambda^2(1-\lambda)\theta^k\geq 0$. When $k$ is odd, we have $\theta^k \leq 0$ and $(N-2)\lambda^2(1-\lambda)\theta^k\leq 0$. Hence, $\theta^k\big[ (N-2)\lambda^2(1-\lambda)\theta^k + (N-2)\lambda^2 +2\lambda -1 \big]\geq \theta^k \big[ (N-2)\lambda^2+2\lambda-1 \big]$, for all $k$. Let $\xi(k,\lambda,\theta,N):= 2(1-\lambda+\lambda\theta^k)^{N-2}\theta^k ( (N-2)\lambda^2+2\lambda-1 )$. Based on above results, it is easy to verify $\xi(2q,\lambda,\theta,N) + \xi(2q+1,\lambda,\theta,N) > 0$, and hence $\sum_{k=2}^{\infty}\xi(k,\lambda,\theta,N) > 0$.

    The remaining term is $\xi(1,\lambda,\theta,N)$ and $(N\lambda-1)$. Summing them, denote $\rho(\theta,\lambda,N):=2\theta[(N-2)\lambda^2(1+\theta)-(N-2)\lambda^3\theta+2\lambda-1] + N\lambda -1$, we have $2(1-\lambda+\lambda\theta)^{N-2}\theta[(N-2)\lambda^2(1+\theta)-(N-2)\lambda^3\theta+2\lambda-1]+ N\lambda - 1 \geq \rho(\theta,\lambda,N)$. The partial derivative of $\rho(\theta,\lambda,N)$ with respect to $\theta$ is $\frac{\partial\rho(\theta,\lambda,N)}{\partial \theta}=2[(N-2)\lambda^2(1+\theta)-(N-2)\lambda^3\theta+2\lambda-1] + 2(N-2)\lambda^2(1-\lambda)\theta$, which is an increasing function of $\theta$ (convex function in $\theta\in(-1,0)$). Choose $\theta=0$, $\frac{\partial\rho(\theta,\lambda,N)}{\partial \theta}|_{\theta=0}=2(N-2)\lambda^2+4\lambda-2 > 0$ when $\lambda\in[\frac{\sqrt{N-1}-1}{N-2},1]$. Choosing $\theta=-1$, the partial derivative becomes $\frac{\partial\rho(\theta,\lambda,N)}{\partial \theta}|_{\theta=-1}=4(N-2)\lambda^3-2(N-2)\lambda^2+4\lambda-2$. 
    
    If $\frac{\partial\rho(\theta,\lambda,N)}{\partial \theta}|_{\theta=-1} < 0$, $\rho(\theta,\lambda,N) \geq \min \{\rho(-1,\lambda,N)+\frac{1}{2}\frac{\partial\rho(\theta,\lambda,N)}{\partial \theta}|_{\theta=-1}, \rho(0,\lambda,N)-\frac{1}{2}\frac{\partial\rho(\theta,\lambda,N)}{\partial \theta}|_{\theta=0}\}$. Note $\rho(-1,\lambda,N)+\frac{1}{2}\frac{\partial\rho(\theta,\lambda,N)}{\partial \theta}|_{\theta=-1} = -2(N-2)\lambda^3-4\lambda+2+N\lambda-1+2(N-2)\lambda^3-(N-2)\lambda^2+2\lambda-1=(N-2)\lambda(1-\lambda)\geq 0$ and $\rho(0,\lambda,N)-\frac{1}{2}\frac{\partial\rho(\theta,\lambda,N)}{\partial \theta}|_{\theta=0} = N\lambda-1-(N-2)\lambda^2-2\lambda+1 = (N-2)\lambda(1-\lambda)\geq 0$. Thus, $\rho(\theta,\lambda,N)\geq 0$ in this case.
    
    If $\frac{\partial\rho(\theta,\lambda,N)}{\partial \theta}|_{\theta=-1}=4(N-2)\lambda^3-2(N-2)\lambda^2+4\lambda-2=2(2\lambda-1)((N-2)\lambda^2+1)\geq 0$, the minimum $\rho(\theta,\lambda,N)$ is achieved at smallest $\theta$ and $\lambda\geq \frac{1}{2}$. Since $\lambda\geq \frac{1}{2}$, $\theta\in [\frac{\lambda-1}{\lambda},0]$, and hence $\rho(\theta,\lambda,N)\geq \rho(\frac{\lambda-1}{\lambda},\lambda,N)=-2(\lambda-1)(N-2)\lambda^2+6(N-2)\lambda(\lambda-1)-2(N-4)(\lambda-1)-\frac{2(\lambda-1)}{\lambda}+N\lambda-1$, and $-\frac{2(\lambda-1)}{\lambda}\geq -2$. Thus, $\rho(\theta,\lambda,N)\geq-2\lambda^3(N-2)+8(N-2)\lambda^2+(20-7N)\lambda+2N-11$. Taking derivative, we have $-6(N-2)\lambda^2+16(N-2)\lambda+20-7N$, which is positive when $\lambda=1$. The derivative of $-6(N-2)\lambda^2+16(N-2)\lambda+20-7N$ is $(N-2)(17-12\lambda)\geq 0$. Combined with $\lambda\in[\frac{1}{2},1]$, the value of $2\lambda^3(N-2)+8(N-2)\lambda^2+(20-7N)\lambda+2N-11$ is larger than the value of choosing $\lambda=\frac{1}{2}$ (when $N < 30)$ or the value of $\lambda=\frac{1}{2}$ plus the derivative at $\lambda=\frac{1}{2}$ times $\frac{1}{2}$ (when $N\geq 30$). The value at $\lambda=\frac{1}{2}$ is $\frac{9}{4}N-\frac{9}{2}\geq 0$. When $N\geq 30$, the value of $\lambda=\frac{1}{2}$ plus the derivative at $\lambda=\frac{1}{2}$ times $\frac{1}{2}$ is $2N+3>0$. Thus, $\rho(\theta,\lambda,N)\geq 0$ in this case.
    
    As $\rho(\theta,\lambda,N) \geq 0$ in above all cases, and hence $\sum_{k=2}^{\infty}\xi(k,\lambda,N) > 0$.

    Next, we consider $(N-2)\lambda+2\lambda-1 < 0$. Solving the inequation, we have $\lambda\in[\frac{1}{N},\frac{\sqrt{N-1}-1}{N-2})$. Note when $|(N-2)\lambda^2(1-\lambda)\theta^k|\geq|(N-2)\lambda^2+2\lambda-1|$, $\theta^k((N-2)\lambda^2(1-\lambda)\theta^k+(N-2)\lambda^2+2\lambda-1)\geq 0$ for both odd and even $k$, and hence the summation is nonnegative for these $k$. So, we only need to analyze the case $|(N-2)\lambda^2(1-\lambda)\theta^k|\geq|(N-2)\lambda^2+2\lambda-1|$, in which $\omega(\lambda,\theta,k,N)\geq 0$ when $k$ is odd and $\omega(\lambda,\theta,k,N)\leq 0$ when $k$ is even. Therefore, we consider $\omega(\lambda,\theta,2q-1,N)+\omega(\lambda,\theta,2q,N)$. Note $\theta^{2q}\leq -\frac{(N-2)\lambda^2+2\lambda-1}{(N-2)\lambda^2(1-\lambda)}$ in this case.

    Recall the function $\xi(k,\lambda,\theta,N)$, we have $\omega(\lambda,\theta,2q-1,N)+\omega(\lambda,\theta,2q,N)\geq(1-\lambda)^{N-2}\xi(2q-1,\lambda,\theta,N)+\xi(2q,\lambda,\theta,N) = (1-\lambda)^{N-2}\theta^{2q-1}[(N-2)\lambda^2+2\lambda-1][(1-\lambda+\lambda\theta^{2q-1})^{N-2}+\theta (1-\lambda+\lambda\theta^{2q})^{N-2}]$. Note $\theta^{2q-1}\leq 0$ and $(N-2)\lambda^2+2\lambda-1\leq 0$. The remaining is to compare $(1-\lambda+\lambda\theta^{2q-1})^{N-2}$ and $-\theta (1-\lambda+\lambda\theta^{2q})^{N-2}$. Divide these two terms, we have $\frac{(1-\lambda+\lambda\theta^{2q-1})^{N-2}}{-\theta (1-\lambda+\lambda\theta^{2q})^{N-2}}\geq\frac{1}{-\theta}\big(\frac{1-\lambda+\lambda\theta^1}{1-\lambda+\lambda\theta^2}\big)^{N-2}$.
    
    When $N\geq 13$, we consider the bound $\theta\in[-\frac{\lambda}{1-\lambda},0]$ and $\lambda\in[\frac{1}{N},\frac{\sqrt{N-1}-1}{N-2}]$, and hence $\frac{(1-\lambda+\lambda\theta^{2q-1})^{N-2}}{-\theta (1-\lambda+\lambda\theta^{2q})^{N-2}}\geq\frac{1-\lambda}{\lambda}(\frac{1-\lambda-\lambda^2/(1-\lambda)}{1-\lambda+\lambda^3/(1-\lambda)^2})^{N-2}=\frac{1-\lambda}{\lambda}(\frac{1-3\lambda+2\lambda^2}{1-3\lambda+3\lambda^2})^{N-2}=\frac{1-\lambda}{\lambda}(1-\frac{1}{1/\lambda^2-3/\lambda+3})^{N-2}$.  By doing the partial derivative with respect to $\lambda$, we have $\frac{\partial (1-\lambda)/\lambda}{\partial \lambda} = \frac{-1}{\lambda^2} < 0$ and $\frac{\partial (1/\lambda^2-3/\lambda)}{\partial \lambda}=1/\lambda^3(3\lambda-2)<0$ when $\lambda\in[\frac{1}{N},\frac{\sqrt{N-1}-1}{N-2}]$. Thus, $ \frac{1-\lambda}{\lambda}(1-\frac{1}{1/\lambda^2-3/\lambda+3})^{N-2}$ is a decreasing function of $\lambda$. Taking $\lambda=\frac{\sqrt{N-1}-1}{N-2}$ into this function, we obtain the lower bound, which is a function of only one variable $N$. By doing the partial derivative, it can be verified that it is an increasing function of $N$. Taking $N=13$, we can verify $\frac{1-\lambda}{\lambda}(1-\frac{1}{1/\lambda^2-3/\lambda+3})^{N-2}\geq 1$. Hence $\frac{(1-\lambda+\lambda\theta^{2q-1})^{N-2}}{-\theta (1-\lambda+\lambda\theta^{2q})^{N-2}}\geq 1$ for $N\geq 13$.

    When $N<13$, we should consider both bound for $\theta$, $\theta\in[-\frac{\lambda}{1-\lambda},0]$ and $\theta^{2q}\leq -\frac{(N-2)\lambda^2+2\lambda-1}{(N-2)\lambda^2(1-\lambda)}$, which comes from the inequation condition $\theta^k((N-2)\lambda^2(1-\lambda)\theta^{2i}+(N-2)\lambda^2+2\lambda-1)\leq 0$. For $\lambda$, we separate its domain to be $[\frac{1}{N}, \Gamma]$ and $[\Gamma, \frac{\sqrt{N-1}-1}{N-2}]$, where $\Gamma=0.03*\frac{1}{N}+0.97*\frac{\sqrt{N-1}-1}{N-2}$. It can be verified that, when $\lambda\in[\frac{1}{N}, \Gamma]$, $\frac{1-\lambda}{\lambda}(1-\frac{1}{1/\lambda^2-3/\lambda+3})^{N-2}\geq 1$ for every $N$ that $N<13$, using the same strategy as that for $N\geq 13$ and the only change is the upper bound for $\lambda$. When $\lambda\in[\Gamma, \frac{\sqrt{N-1}-1}{N-2}]$, we analyze the sign of $(1-\lambda+\lambda\theta^{2q-1})^{N-2}+\theta (1-\lambda+\lambda\theta^{2q})^{N-2}$. Since $\theta^{2q}\leq -\frac{(N-2)\lambda^2+2\lambda-1}{(N-2)\lambda^2(1-\lambda)}$, by doing the partial derivative, the upper bound of $\theta^{2q}$ is a decreasing function of $\lambda$. Hence, the maximum $\theta^{2q}$ is achieved when $\lambda=\Gamma$ and then $\theta^{2q}\leq -\gamma$, where $\gamma=\frac{(N-2)\Gamma^2+2th-1}{(N-2)\Gamma^2(1-\Gamma)}$ is a variable only related to $N$. Thus, $(1-\lambda+\lambda\theta^{2q-1})^{N-2}+\theta (1-\lambda+\lambda\theta^{2q})^{N-2}\geq(1-\lambda-\lambda\sqrt{\gamma})^{N-2}-\sqrt{\gamma}(1-\lambda+\lambda\gamma)^{N-2}$, which, verified by doing partial derivative, is a decreasing function of $\lambda$. Then, taking $\lambda=\frac{\sqrt{N-1}-1}{N-2}$, we can verify that $(1-\lambda-\lambda\sqrt{\gamma})^{N-2}-\sqrt{\gamma}(1-\lambda+\lambda\gamma)^{N-2}\geq 0$ for every $N$ that $N<13$. Therefore, $(1-\lambda+\lambda\theta^{2q-1})^{N-2}+\theta (1-\lambda+\lambda\theta^{2q})^{N-2}\geq 0$ when $N<13$

    As $(1-\lambda+\lambda\theta^{2q-1})^{N-2}+\theta (1-\lambda+\lambda\theta^{2q})^{N-2}\geq 0$ for both $N\geq 13$ and $N<13$, $\omega(\lambda,\theta,2q-1,N)+\omega(\lambda,\theta,2q,N) \geq 0$, and hence $\sum_{k=1}^{\infty}\xi(k,\lambda,\theta,N) \geq 0$ when $(N-2)\lambda+2\lambda-1 < 0$.
    
    Since $\sum_{k=1}^{\infty}\xi(k,\lambda,\theta,N) \geq 0$ when $(N-2)\lambda+2\lambda-1 < 0$ and $\sum_{k=1}^{\infty}\xi(k,\lambda,\theta,N)+N\lambda-1 \geq 0$ when $(N-2)\lambda+2\lambda-1 \geq 0$, we have, when $\lambda \geq \frac{1}{N}$,
    \begin{align*}
        \frac{\partial (v_a^2/m_a^2)}{\partial \lambda} \geq 0,
    \end{align*}



    Now, we discuss passive users. Based on Theorem~\ref{thm:mean_vairance}, we have
    \begin{align*}
         & \frac{\partial (v_p^2/m_p^2)}{\partial \lambda} = 2\sum_{k=1}^{\infty} \Big[ \frac{CN(\theta^k-1)(1-\lambda+\lambda\theta^k)^{CN-1}m_p}{m_p^2} \\
         & + \frac{CN(1-\lambda)^{CN-1}(1-\lambda+\lambda\theta^k)^{CN}}{m_p^2} \Big] + \frac{CN(1-\lambda)^{CN}}{m_p^2} \\
         & = \frac{2CN(1-\lambda)^{CN-1}[1-\lambda + \sum_{k=1}^{\infty}(1-\lambda+\lambda\theta^k)^{CN-1}\theta^k]}{m_p^2}.
    \end{align*}
    Since $\theta \in (-1,0)$ and $1-\lambda+\lambda\theta^k \geq 0$ for all $k$,
    \begin{align*}
         & \frac{\partial (v_p^2/m_p^2)}{\partial \lambda} \\
         & \geq \frac{2CN(1-\lambda)^{CN-1}[1-\lambda + \sum_{k=1}^{\infty}(1-\lambda)^{CN-1}\theta^k]}{m_p^2} \\
         & = \frac{CN(1-\lambda)^{CN-1}[1-\lambda + (1-\lambda)^{CN-1} \frac{\theta}{1-\theta}]}{m_p^2}.
    \end{align*}
    When $\lambda\in [0,\frac{1}{2}]$, $\theta\in [-\frac{\lambda}{1-\lambda},0]$, and hence
    \begin{align*}
        & \frac{\partial (v_p^2/m_p^2)}{\partial \lambda} \geq \frac{2CN(1-\lambda)^{CN-1}[1-\lambda - (1-\lambda)^{CN-1} \lambda]}{m_p^2} \\
        & = \frac{2CN(1-\lambda)^{CN-1}[1-(1+(1-\lambda)^{CN-1})\lambda]}{m_p^2} \\
        & \geq \frac{2CN(1-\lambda)^{CN-1}[1-(1+1)\frac{1}{2}]}{m_p^2} = 0.
    \end{align*}
    When $\lambda\in [\frac{1}{2},1]$, $\theta\in [\frac{\lambda-1}{\lambda}, 0]$, and then
    \begin{align*}
        & \frac{\partial (v_p^2/m_p^2)}{\partial \lambda} \geq \frac{2CN(1-\lambda)^{CN-1}[1-\lambda - (1-\lambda)^{CN}]}{m_p^2} \\
        & \geq \frac{2CN(1-\lambda)^{CN-1}[1-\lambda - (1-\lambda)]}{m_p^2} = 0.
    \end{align*}

    Therefore, when $\lambda\geq \frac{1}{N}$,
    \begin{align*}
        \frac{\partial (v_p^2/m_p^2)}{\partial \lambda} \geq 0.
    \end{align*}
    
\subsection{Proof of Lemma \ref{lemma:opt_theta}}
\label{app:active}
    The partial derivative of $v_a^2$ with respect to $\theta$ when $\lambda$ is determined is
    \begin{align*}
        & \frac{\partial v_a^2}{\partial \theta} = 2m_a\sum_{k=1}^{\infty}[(1-\lambda)k\theta^{k-1}(1-\lambda+\lambda \theta^{k})^{N-1}  \\
        & \quad +(\lambda+(1-\lambda)\theta^k)(N-1)\lambda k \theta^{k-1}(1-\lambda+\lambda \theta^k)^{N-2})]  \\
        = & 2m_a \sum_{k=1}^{\infty}[1-2\lambda + N\lambda^2+ N\lambda(1-\lambda)\theta^{k}] \\
        & \quad \times k\theta^{k-1}(1-\lambda+\lambda \theta^{k})^{N-2}.
    \end{align*}
    Let $\psi(k):= 2m_a [1-2\lambda + N\lambda^2+ N\lambda(1-\lambda)\theta^{k}]k\theta^{k-1}(1-\lambda+\lambda \theta^{k})^{N-2}$ be the $k$th term in the above equation. 

    We first show that $\alpha,\beta < \frac{1}{2}$. Choosing $y=0$ and $y=\frac{1}{2}$, we have $h_N(0)=\Bar{h}_N(0)=1$, $h_N(\frac{1}{2})=-\frac{3}{8}N+\frac{1}{12}<0$ and $\Bar{h}_N(\frac{1}{2})=-\frac{1}{8}N < 0$. Therefore, there must exists one root that in $(0,\frac{1}{2})$ for both $h_N(y)$ and $\Bar{h}_N(y)$. Thus, $\lambda\in[0,\frac{1}{2})$ and hence $\theta\in[-\frac{\lambda}{1-\lambda},1]$.
    
    Therefore, $[1-2\lambda + N\lambda^2+ N\lambda(1-\lambda)\theta^{k}] \geq 1-2\times\frac{1}{2} + N\lambda^2 - N \lambda (1-\lambda) \frac{\lambda}{1-\lambda}=0$, and $1-\lambda+\lambda \theta^{k}\geq 1-\lambda+\lambda \theta^{1}\geq 1-\lambda-\lambda\frac{\lambda}{1-\lambda} > 1-2\lambda \geq 0$. So, $\psi(k)<0$ only when $\theta<0$ and $k$ is even.

The expression of $\frac{\psi(2i-1)}{\psi(2i)}$ is
\begin{align*}
    -\frac{\psi(2i-1)}{\psi(2i)} =& -\frac{1-2\lambda + N\lambda^2+ N\lambda(1-\lambda)\theta^{2i-1}}{1-2\lambda + N\lambda^2+ N\lambda(1-\lambda)\theta^{2i}}\frac{2i-1}{2i} \\
    &\times \frac{\theta^{2i-1}}{\theta^{2i}}\Big(\frac{1-\lambda+\lambda \theta^{2i-1}}{1-\lambda+\lambda \theta^{2i}}\Big)^{N-2}.
\end{align*}

When $\theta \geq 0$, it is easy to verify that $-\frac{\psi(2i-1)}{\psi(2i)}\geq 1$ since $\theta\leq 1$.

When $\frac{-\lambda}{1-\lambda}\leq \theta<0$, we have
\begin{align*}
    &\frac{\psi(2i-1)}{\psi(2i)} \geq \frac{1-2\lambda + N\lambda^2+ N\lambda(1-\lambda)\theta^{1}}{1-2\lambda + N\lambda^2+ N\lambda(1-\lambda)\theta^{2}}\frac{1}{2}\frac{\theta^{1}}{\theta^{2}} \\
    &\quad \times \Big(\frac{1-\lambda+\lambda \theta^{1}}{1-\lambda+\lambda \theta^{2}}\Big)^{N-2} \\
    & \geq \frac{1-2\lambda + N\lambda^2 - N\lambda(1-\lambda) \frac{\lambda}{1-\lambda}}{1-2\lambda + N\lambda^2+ N\lambda(1-\lambda)\frac{\lambda^2}{(1-\lambda)^2}} \frac{1}{2} \frac{1-\lambda}{\lambda} \\
    & \quad\times \Big(\frac{1-\lambda-\lambda \frac{\lambda}{1-\lambda}}{1-\lambda+\lambda \frac{\lambda^2}{(1-\lambda)^2}}\Big)^{N-2}\\
    & = \frac{1-\lambda}{2\lambda} \frac{1-3\lambda + 2\lambda^2}{1-3\lambda+(N+2)\lambda^2} (1-\frac{\lambda^2}{1-3\lambda+3\lambda^2})^{N-2}.
\end{align*}

Define $A_N(y), \Bar{B}(y)$ and $C_N(y)$ as follows:
\begin{align*}
    A_N(y) &= \frac{1-3y + 2y^2}{1-3y+(N+2)y^2} (1-\frac{y^2}{1-3y+3y^2})^{N-2}, \\
    \Bar{B}(y) &= \frac{2y}{1-y}, \quad  C_N(y) = \frac{1-3y-(N-4)y^2}{1-3y+(N+2)y^2}.
\end{align*}
Then, $\frac{\psi(2i-1)}{\psi(2i)} \geq \frac{A_N(y)}{\Bar{B}(y)}$.

First, we prove $C_N(y) \geq \Bar{B}(y)$. Subtracting $\Bar{B}(y)$ from $C_N(y)$, we have
\begin{align*}
    &C_N(y) - \Bar{B}(y)  = \frac{1-3y-(N-4)y^2}{1-3y+(N+2)y^2} - \frac{2y}{1-y} \\
    & = \frac{(1-3y)^2-(1-y)(N-4)y^2-2y(N+2)y^2}{(1-3y+(N+2)y^2)(1-y)} \\
    & = \frac{h_N(y)}{(1-3y+(N+2)y^2)(1-y)}.
\end{align*}

As $\alpha$ is the smallest positive root of $h_N(y)$ and $h_N(0)=1>0$, $h_N(y) \geq 0$ when $y\in[0,\alpha]$. Note $\alpha < \frac{1}{2}$.

Now we check whether $(1-3y+(N+2)y^2)$ is positive. The derivative of $(1-3y+(N+2)y^2)$ is $2(N+2)y-3$ and hence it achieves minimum when $y=\frac{3}{2(N+2)}$. Choosing $y=\frac{3}{2(N+2)}$, we have $(1-3y+(N+2)y^2)=1-\frac{9}{2(N+2)}+\frac{9}{4(N+2)}= 1-\frac{9}{4(N+2)}\geq 1-\frac{9}{4(2+2)}>0$.

Since $(1-3y+(N+2)y^2)$, $(1-y)$ and $\Bar{h}_N(y)$ are all non-negative when $y\in[0,\alpha]$, we have 
\begin{align*}
    C_N(y) - \Bar{B}(y) = \frac{h_N(y)}{(1-3y+(N+2)y^2)(1-y)} \geq 0.
\end{align*}

Next, we prove $A_N(y)\geq C_N(y)$.

We know that $(1-y)^N \approx 1-Ny$ from binomial approximation. Let $j(y)=(1-y)^N-(1-Ny)$. The derivative of $j(y)$ is $N [ 1 - (1-y)^{N-1}]$. As $(1-y)^{N-1} \leq 1$ when $0\leq y \leq 1$, we have $j(y)$ is increasing when $y\in[0,1]$. Hence, $j(y) \geq j(0) = 0$, which means $(1-y)^N \geq 1-Ny$, when $y\in[0,1]$.

Now, we want to find whether $\frac{y^2}{3y^2-3y + 1}$ is between $0$ and $1$. The derivative of $3y^2-3y + 1 - y^2$ is $4y-3$. Thus, when $0\leq y\leq \frac{1}{2}$, $3y^2-3y + 1 - y^2 \geq 3\times\frac{1}{4}-3\times\frac{1}{2}+1-\frac{1}{4}=0$. Therefore, $3y^2-3y + 1 \geq y^2$. As $3y^2-3y+1>0$ and $y^2 \geq 0$, we have $\frac{y^2}{3y^2-3y + 1} \leq 1$ and $\frac{y^2}{3y^2-3y + 1} \geq 0$ when $y\in[0,\frac{1}{2}]$.

Thus, when $y\in[0,\frac{1}{2}]$,
\begin{align*}
   (1 - \frac{y^2}{3y^2 - 3y + 1})^{N-2} \geq 1 - (N-2)\frac{y^2}{3y^2-3y + 1}.
\end{align*}

Then,
\begin{align*}
    A_N(y) & = \frac{1-3y + 2y^2}{1-3y+(N+2)y^2} (1-\frac{y^2}{1-3y+3y^2})^{N-2} \\
    & \geq \frac{1-3y + 2y^2}{1-3y+(N+2)y^2} [1-(N-2)\frac{y^2}{1-3y+3y^2}] \\
    & = \frac{1-3y + 2y^2}{1-3y+(N+2)y^2}\frac{1-3y-(N-5)y^2}{1-3y+3y^2}.
\end{align*}
When $N\geq 2$, $1-3y-(N-5)y^2 \leq 1-3y+3y^2$, and hence
\begin{align*}
    \frac{1-3y-(N-5)y^2}{1-3y+3y^2} \geq \frac{1-3y-(N-4)y^2}{1-3y+2y^2}.
\end{align*}
Therefore,
\begin{align*}
    A_N(y) &\geq \frac{1-3y + 2y^2}{1-3y+(N+2)y^2}\frac{1-3y-(N-4)y^2}{1-3y+2y^2} \\
    &= \frac{1-3y-(N-4)y^2}{1-3y+(N+2)y^2} = C_N(y).
\end{align*}

So, when $y\in[0,\alpha]$, $A_N(y) \geq C_N(y) \geq \Bar{B}(y)$. Thus,
\begin{align*}
    \frac{\psi(2i-1)}{\psi(2i)} \geq \frac{A_N(\lambda)}{\Bar{B}(\lambda)} \geq 1.
\end{align*}

Hence, $\frac{\partial v_a^2}{\partial \theta} \geq 0$ and $v_a^2$ achieves its minimum value when $\theta$ is minimized. Thus, choosing $r=\frac{\lambda}{1-\lambda}$ and $s=1$ minimizes $v_a^2$. 

Since $E[AoI_a^z]$ is an increasing function of $v_a^2$, choosing $r=\frac{\lambda}{1-\lambda}$ and $s=1$ also minimizes $E[AoI_a^z]$.

Next, we consider passive users.

Recall the expression of $v_p^2$,
\begin{align*}
    v_p^2 &= 2\sum_{k=1}^{\infty}((1-\lambda +\lambda \theta^{k})^{CN} -m_p) m_p +  m_p - m_p^2.
\end{align*}
For this expression, it is easy to verify that $v_p^2|_{\theta>0} \geq v_p^2|_{\theta=0}$. Therefore, $v_p^2$ reaches its minimum in the region $-\frac{\lambda}{1-\lambda} \leq \theta \leq 0$.

The partial derivative of $v_p^2$ with respect to $\theta$ is
\begin{align*}
    \frac{\partial v_p^2}{\partial \theta} = 2m_p \sum_{k=1}^{\infty}CN k \lambda \theta^{k-1} (1 - \lambda + \lambda \theta^k)^{CN-1}.
\end{align*}
Let $\phi(k)=2 m_p CN k \lambda \theta^{k-1} (1 - \lambda + \lambda \theta^k)^{CN-1}$ be the $k-th$ term of the partial derivative above. When $\theta \in [-\frac{\lambda}{1-\lambda},0]$, we can find that $\phi(k) \geq 0$ when $k$ is odd and $\phi(k) < 0$ when $k$ is even. Thus, $\frac{\partial v_p^2}{\partial \theta}$ is positive if $-\frac{\phi(2i-1)}{\phi(2i)} > 1$ for all positive integer $i$.

Plugging in the formula of $\phi(k)$,
\begin{align*}
    -\frac{\phi(2i-1)}{\phi(2i)} = - \frac{2i-1}{2i}\frac{1}{\theta} \Big( \frac{1 - \lambda + \lambda \theta^{2i-1}}{1 - \lambda + \lambda \theta^{2i}} \Big)^{CN-1}.
\end{align*}

Since $\theta^{2i-1}\geq \theta \geq -\frac{\lambda}{1-\lambda}$ and $\theta^{2i}\leq \theta^{2} \leq \frac{\lambda^2}{(1-\lambda)^2}$ when $\theta \in [-\frac{\lambda}{1-\lambda},0]$, we have
\begin{align*}
    -\frac{\phi(2i-1)}{\phi(2i)} &\geq -\frac{\phi(1)}{\phi(2)} = -\frac{1}{2} \frac{1}{\theta} \Big( \frac{1-\lambda+\lambda \theta}{1-\lambda+\lambda \theta^{2}} \Big)^{CN-1} \\
    & \geq \frac{1-\lambda}{2\lambda} \Big( 1 - \frac{\lambda^2}{3\lambda^2-3\lambda+1} \Big)^{CN-1}.
\end{align*}

Define $\Bar{A}_{CN}(y)$ as follows.
\begin{align*}
    \Bar{A}_{CN}(y) &= \Big( 1 - \frac{y^2}{3y^2-3y+1} \Big)^{CN-1}.
\end{align*}

When $y\in[0, \frac{1}{2}]$,
\begin{align*}
    &\Bar{A}_{CN}(y) - \Bar{B}(y) = \Big( 1 - \frac{y^2}{3y^2-3y+1} \Big)^{CN-1} - \frac{2y}{1-y} \\
    & \geq 1- (CN-1)\frac{y^2}{3y^2-3y+1} - \frac{2y}{1-y} \\
    & = \frac{\Bar{h}_{CN}(y)}{(3y^2-3y+1)(1-y)}.
\end{align*}

Note $\Bar{h}_{CN}(0) = 1$, $\Bar{h}_{CN}(\frac{1}{2}) = -\frac{CN}{8}$, and $\beta$ is the smallest positive root of $\Bar{h}_{CN}(y)$. Thus, $\Bar{h}_{CN}(y)\geq 0$ if $y\in[0,\beta]$. When $y\in[0,\beta]$, $3y^2-3y+1 > 0$ and $1-y>0$, and hence
\begin{align*}
    \Bar{A}_{CN}(y) - \Bar{B}(y) = \frac{\Bar{h}_{CN}(y)}{(3y^2-3y+1)(1-y)} \geq 0.
\end{align*}

Therefore, if $\lambda \in [0,\beta]$,
\begin{align*}
    -\frac{\phi(2i-1)}{\phi(2i)} \geq \frac{\Bar{A}_{CN}(\lambda)}{\Bar{B}(\lambda)} \geq 1.
\end{align*}

So, $\frac{\partial v_p^2}{\partial \theta} > 0$ and $v_p^2$ achieves its minimum value when $\theta$ is minimized and $\lambda \in [0,\beta]$. Thus, choosing $r=\frac{\lambda}{1-\lambda}$ and $s=1$ minimizes $v_p^2$. 


\subsection{Proof of Lemma \ref{lemma_alpha}}
\label{app:lemma_a}

Let us rewrite the expression of $h_N(y)$.
\begin{align*}
    h_N(y) = -(N+8)y^3 - (N-13) y^2 - 6y + 1.
\end{align*}

The derivative of $h_N(y)$ is
\begin{align*}
    h'_N(y) = -3(N+8)y^2-2(N-13)y-6.
\end{align*}
Let $h'_N(y) = 0$, we have $y_1=\frac{2(N-13) + \sqrt{4(N-13)^2-72(N+8)}}{-6(N+8)}$ and $y_2=\frac{2(N-13) - \sqrt{4(N-13)^2-72(N+8)}}{-6(N+8)}$. Thus, the roots of $h'_N(y)$ are either not existing or all negative. Since there is no root when $y \geq 0$ and $h'_N(0)=-6<0$, $h'_N(y)$ is always negative when $y\geq 0$. Therefore, $h_N(y)$ is always decreasing when $y\geq 0$. Considering $h_N(0)=1>0$ and $h_N(1)=-2N<0$, $\alpha$ must exist between $0$ and $1$, and $h_N(y)>0$ if and only if $y<\alpha$, when $y \geq 0$.

Let $y=\frac{1}{N}$, we have
\begin{align*}
    h_N(\frac{1}{N}) & = \frac{-(N+8)}{N^3} - \frac{(N-13)}{N^2} - \frac{6}{N} + 1\\
    & = \frac{N^3-7N^2+12N-8}{N^3}.
\end{align*}
Solving $N^3-7N^2+12N-8=0$, we obtain the only real root, which is approximately $4.8751$. Since $N^3-7N^2+12N-8\rightarrow \infty$ when $N\rightarrow \infty$, we have $N^3-7N^2+12N-8>0$ when $N\geq 5$.

Therefore, $h_N(\frac{1}{N}) > 0$, and hence $\alpha > \frac{1}{N}$ when $N>4$.

Now we proof $\beta>\frac{1}{N}$ when $N>C+4$. Let us rewrite $\Bar{h}_N(y)$ here,
\begin{align*}
    \Bar{h}_{CN}(y) = (CN-10)y^3-(CN-13) y^2 -6y +1.
\end{align*}

The analysis of $\Bar{h}_{CN}(y)$ depends on the sign of $CN-10$.

When $CN-10=0$, $\Bar{h}_{CN}(y)=3y^2-6y+1$. There are two roots for $\Bar{h}_{CN}(y)$, which are $\frac{3-\sqrt{6}}{3}\approx 0.1835$ and $\frac{3+\sqrt{6}}{3} \approx 1.8165$. Note $N$ must be $10$ when $CN=10$ and $N > C+4$. Thus, $\beta \approx 0.1835 > \frac{1}{N} = 0.1$.

When $CN-10 < 0$, the derivative of $\Bar{h}_{CN}(y)$ is
\begin{align*}
    \Bar{h}'_{CN}(y) = 3(CN-10)y^2 - 2(CN-13)y -6.
\end{align*}
Let $\Bar{h}'_{CN}(y) = 0$, we obtain two roots, one is $x_1 = \frac{2(CN-13) + \sqrt{4(CN-13)^2 + 72(CN-10)}}{6(CN-10)}$, and the other one is $x_2 = \frac{2(CN-13) - \sqrt{4(CN-13)^2 + 72(CN-10)}}{6(CN-10)}$. Look at the term in the square root, $4(CN-13)^2 + 72(CN-10)= 4C^2N^2 - 32CN -44$. It has one negative $4-3\sqrt{3} \approx -1.1962$ and one positive root $4 + 3 \sqrt{3}\approx 9.1692$. Thus, $\Bar{h}'_{CN}(y)$ does not have real roots when $1 < CN < 10$, and hence $\Bar{h}_{CN}(y)$ is always a decreasing function in this case. Therefore, $\Bar{h}_{CN}(y)$ only has one root, which is $\beta$. In this case, $\Bar{h}_{CN}(y)>0$ if and only if $y<\beta$.

When $CN-10 > 0$, $\Bar{h}_{CN}(-\infty) < 0$, $\Bar{h}_{CN}(0)=1>0$, $\Bar{h}_{CN}(1)=-2<0$, and $\Bar{h}_{CN}(\infty) > 0$. Thus, there must exist three roots, one is less than $0$, one is between $0$ and $1$, and one is larger than $1$. Then, there is just one root, $\beta$, between $0$ and $1$, and $\Bar{h}_{CN}(y) > 0$ if and only if $y<\beta$ when $y\in(0,1)$.

Therefore, if $\Bar{h}_{CN}(y_1)>0$ for a certain $y_1\in(0,1)$, we have $\beta>y_1$ when $1\leq CN < 10$ or $CN > 10$. Taking $y=\frac{1}{N}$ in $\Bar{h}_{CN}(y)$, we have 
\begin{align*}
    \Bar{h}_{CN}(\frac{1}{N}) &= \frac{CN-10}{N^3} - \frac{CN-13}{N^2} - \frac{6}{N} + 1 \\
    & = \frac{CN-10 - CN^2 + 13 N - 6N^2 +N^3 }{N^3} \\
    & = \frac{N^3-(6+C)N^2+(13+C)N-10}{N^3}.
\end{align*}
When $N > C+4$,
\begin{align*}
    \Bar{h}_{CN}(\frac{1}{N}) \geq& \frac{(C+5)^3-(6+C)(C+5)^2+(13+C)(C+5)}{N^3} \\
    &  - \frac{10}{N^3} =  \frac{8C+30}{N^3} > 0.
\end{align*}

Therefore, $\beta > \frac{1}{N}$ when $N>C + 4$.

\end{document}